\documentclass{article}

\usepackage{amssymb}
\usepackage{amsmath}
\usepackage{bbm}
\usepackage{latexsym}
\usepackage{xspace}
\usepackage{slashbox}
\usepackage{gastex}
\usepackage{enumitem}
\usepackage{setspace}

\usepackage{graphicx}  
\usepackage[graphicx]{realboxes}
\usepackage{tikz}
\usetikzlibrary{automata,positioning}
\usepackage{tikz-qtree}
\usepackage{float}
\usepackage{floatflt}

\newtheorem{theorem}{Theorem}

\newtheorem{lemma}[theorem]{Lemma}
\newtheorem{remark}[theorem]{Remark}
\newtheorem{example}[theorem]{Example}
\newtheorem{corollary}[theorem]{Corollary}
\newtheorem{definition}[theorem]{Definition}

\def\squarebox#1{\hbox to #1{\hfill\vbox to #1{\vfill}}}
\newcommand{\qed}{\hspace*{\fill}
       \vbox{\hrule\hbox{\vrule\squarebox{.667em}\vrule}\hrule}\smallskip}
\newenvironment{proof}{\begin{trivlist}
\item[\hspace{\labelsep}{\bf\noindent Proof: }]
}{\qed\end{trivlist}}

\newcommand{\short}[1]{}
\newcommand{\set}[1]{{\left\{#1\right\}}}
\newcommand{\NN}{\mathbbm{N}}

\newcommand{\NP}{\textbf{NP}\xspace}
\newcommand{\NPh}{\textbf{NP-Hard}\xspace}

\newcommand{\PTIME}{\textbf{P}\xspace}

\newcommand{\tup}[1]{\langle #1  \rangle}

\newcommand{\cost}{\mathsf{cost}}
\newcommand{\NEset}{\mathsf{NE}_\mathsf{PO}}

\mathchardef\mhyphen="2D

\newcommand{\ol}[1]{{\widetilde{#1}}}
\newcommand{\util}{\mathsf{util}}
\newcommand{\SO}{\mathsf{SO}}
\newcommand{\PoA}{\mathsf{PoA}}
\newcommand{\PoS}{\mathsf{PoS}}

\newcommand{\todo}[1]{\par\noindent{\raggedright\textsf{TODO$\uparrow$: #1}
		\par\marginpar{\Large \bf $\star$}}}

\newcommand{\stam}[1]{}

\renewcommand{\Game}{{\cal G}}
\newcommand{\Act}{\mathsf{Act}}
\newcommand{\outcome}{\mathsf{outcome}}
\newcommand{\coal}{\overline{\alpha}}
\newcommand{\strat}{{\mathfrak S}}

\begin{document}

\newcommand{\anotefrom}[2]{\textbf{[#1]: {#2}}}
\newcommand{\sanote}[1]{\anotefrom{sa}{#1}}
\newcommand{\ranote}[1]{\anotefrom{ra}{#1}}
\newcommand{\sbnote}[1]{\anotefrom{sb}{#1}}

\title{Equilibria in Quantitative Concurrent Games}
\author{Shaull Almagor$^1$, Rajeev Alur$^2$, and Suguman Bansal$^3$}
\date{\small
	$^1$ Department of Computer Science, Oxford University\\
$^2$ Department of Computer and Information Science, University of Pennsylvania\\
$^3$ Department of Computer Science, Rice University}
\maketitle

\begin{abstract}
Synthesis of finite-state controllers from high-level specifications in 
multi-agent systems can be reduced to solving multi-player concurrent 
games over finite graphs. The complexity of solving such games with 
qualitative objectives for agents, such as reaching a target set, is 
well understood resulting in tools with applications in robotics. In 
this paper, we introduce quantitative concurrent graph games, where 
transitions have separate costs for different agents, and each agent 
attempts to reach its target set while minimizing its own cost along the 
path. In this model, a solution to the game corresponds to a set of 
strategies, one per agent, that forms a Nash equilibrium. We study the 
problem of computing the set of all Pareto-optimal Nash equilibria, and 
give a comprehensive analysis of its complexity and related problems 
such as the price of stability and the price of anarchy. In particular, 
while checking the existence of a Nash equilibrium is NP-complete in 
general, with multiple parameters contributing to the computational 
hardness separately, two-player games with bounded costs on individual 
transitions admit a polynomial-time solution.
\end{abstract}
\section{Introduction}
\label{sec:intro}

The proliferation of massive online protocols such as auctions (Google Auctions, eBay), decentralized crypto-
currencies (Bitcoins), ride-sharing applications (Uber, Lyft), have propelled an interest in the {\em automated} design of {\em provably correct} multi-agent systems. 
{\em Synthesis}, pioneered by Church's Problem~\cite{church1957applications}, is a declarative paradigm for the automated design of provably correct systems. Synthesis is the automated construction of systems from their specifications.  
In the context of multi-agent system, synthesis constructs a {\em controller} that directs agent interactions in order to satisfy the specification. Among others, the synthesis of controllers finds vast application in motion planning in single- and multiple- robot systems~\cite{alur2016compositional,fainekos2009temporal,kress2009temporal,livingston2015cross,wang2016task}.
 

A {\em specification} for the synthesis of controller for 
a multi-agent systems consists two parts: First, a description of individual agent objectives;  Second, a description how agents interact  with each other.
Individual agent objectives are expressed in linear temporal logic over finite domain~\cite{de2015synthesis} as a de-facto in planning, or over infinite domain~\cite{pnueli1977temporal} for liveness and safety properties, and the like. 
Agent interactions are expressed by {\em graph games} in which vertices and edges of a graph denote game states and agent interactions, respectively.
Depending on whether agent interactions occur in turns or concurrently, the edges are graph game are labeled with single agent actions or concurrent agent actions.  
The synthesis of the controller under these specifications corresponds to solving the graph game with agent objectives.

The specifications considered in most existing work is {\em qualitative}. They do not take into account practical aspects such as cost of interaction, amount of resources of agents, and so on.
For example, crucial details in the design of a controller for a multi-robot surveillance task would include the  battery resource consumed by robots in traversing uneven terrains in the environment, total distance/area covered by each and all robots, and so on.  These details cannot be represented qualitatively.

A richer form of specification is {\em quantitative}.
Agent interactions in a graph game should also include 
the {\em quantitative costs} incurred by agents during interactions (e.g. battery consumption). W.l.o.g, costs incurred by different agents along the same transition may differ. In addition, agents may be constrains on their resources (e.g. fixed battery life). Hence, agents may also have the {\em quantiative objective} to optimize their total cost while fulfilling their qualitative objective. 
Synthesis of controllers from quantitative specifications entails solving a quantitative graph games under qualitative and quantitative objectives for its agents. 
Finally, the objective of agents is to optimize its own cost, and not to play against other agents. {\em Nash equilibrium} is a popular choice of solution concept in such non-competitive games~\cite{nisan2007algorithmic}. Therefore, synthesis of controllers in quantitative games with quantitative agent objectives is reduced to Nash equilibria computation in these games. 

Intuitively, Nash equilibria assigns a strategy to each agent such that unilateral deviations by an agent are not beneficial to it~\cite{nash1951non}. 
The computation of Nash equilibrium has been extensively studied for simple one-shot games~\cite{chen2009settling,conitzer2008new,daskalakis2009complexity}, repeated games~\cite{abreu1988theory,andersen2013fast,mailath2006repeated}. 
The problem has been investigated on concurrent graph games~\cite{ummels2015pure} and turn-based quantitative games under reachability objective~\cite{brihaye2010equilibria}. 
Equilibria computation for quantitative concurrent games is open. 

This paper studies Nash equilibria computation for quantitative concurrent graph games with reachability objective. Each agent accumulates its cost until it reaches its set of target states in the graph game, and aims to minimize its cost. We show that determining the existence of Nash equilibria in such games is NP-complete, in general. Our proof argument follows that unilateral deviations for a Nash equilibria by an agent are punished by the coalition of the remaining agents. 
We also study the 
problem of computing the set of all Pareto-optimal Nash equilibria. We prove that the problem is exponential in the worst-case, and identify special cases that admit polynomial-time solutions. 
We also give a comprehensive analysis of the complexity of related problems 
such as the price of stability and the price of anarchy.

\section{Quantitative concurrent graph games}
\label{Sec:QuantGames}

\paragraph{Model description}
A \emph{quantitative concurrent graph game}, abbreviated to {\em game} or \textsf{QCG}, 
is a tuple $\Game = \tup{\Omega, V,  \{\Act_\alpha\}_{\alpha \in \Omega}, v_0, \delta,\cost,F}$, with set of players $\Omega = \{\alpha_1, \dots \alpha_k\},$ set of states $V$, initial state $v_o \in V$. The set of actions for player $ \alpha$ are given by $ \Act_\alpha$. The transition funciton is given by $\delta:V\times\prod_{\alpha \in \Omega} \Act_\alpha \rightarrow V$. 
Alternately, we use $(u, \overline{a}, w  ) \in \delta$ to mean $w  = \delta(u, \overline{a})$.
 The cost function $\cost:V\times \prod_{\alpha\in\Omega}\Act_\alpha \times V \to \NN^\Omega$ assigns a cost vector $(\cost_\alpha)_{\alpha \in \Omega}$ to every transition in $\delta$, where the $\alpha$-th element denotes cost for Player $\alpha$. 
 The target function $F:\Omega\to 2^{V}$ prescribes the target set of states of each player. 

An {\em outcome} of a game $\Game$ is a (finite or infinite) sequence of subsequent transitions beginning in the initial state. Concretely, the sequence $\rho = \tau_0,\tau_1\dots$, where  $\tau_i = (s_i, \vec{a_i}, t_i)\in \delta$ for all $i\geq 0$, is an outcome of the game if $s_0$ is the initial state, and for all $i\geq 0$, $s_{i+1} = t_i$. 
An outcome $\rho$ is said to visit state $s$ if there exists a $j \geq 0$ s.t. $s_j = s$.
An outcome is said to visit a set of states if it visits at least one member state of the set. 

The cost of a player from an outcome is computed by an accumulation of costs incurred by the player along transitions in the outcome. 
In a {\em reachability game}, the cost of a player is the sum of costs incurred along the outcome until its target set is {\em visited for the first time}. 
Formally, the cost of player $\alpha$ in the outcome $\rho$, denoted $\cost_\alpha(\rho)$, is $\Sigma_{j=0}^l \cost_\alpha(\tau_j) $ when $l\geq 0$ is the first index in $\rho$ at which $F(\alpha)$ is visited, and the cost is $\infty$ if $F(\alpha)$ is not visited in $\rho$. 
The objective of each player is to minimize its cost (and in particular, to reach its target). 

\paragraph{Nash equilibrium}
A \emph{strategy} for Player $\alpha$ is a function $\sigma_\alpha:\delta^*\to \Act_\alpha$ which decides the player's next action based on the history of 
transitions taken 
so far\footnote{The careful reader may notice that the history of actions would have sufficed. We choose to work with the history of transitions for cleaner proofs.}.
The set of strategies of player $\alpha$ are denoted by $\strat{\alpha}$.
A strategy $\sigma_\alpha$ is \emph{memoryless} if, intuitively, it prescribes the next action depending only on the current state. That is, if for every two finite outcomes $\rho=\tau_0\cdots\tau_k$ and $\rho'=\tau'_0\cdots\tau'_m$ with $\tau_k=(s_k,\vec{a_k},t_k)$ and $\tau'_m=(s'_m,\vec{b_m},t'_m)$, if $t_k=t'_m$, then $\sigma_\alpha(\rho)=\sigma_\alpha(\rho')$. A memoryless strategy can thus be defined as $\sigma_\alpha:V\to \Act_\alpha$.

A \emph{profile} is a tuple of strategies $P=\tup{\sigma_\alpha}_{\alpha\in \Omega}$, where $\sigma_\alpha$ denotes a strategy for player $\alpha$. 
The profile $P$ induces an outcome, denoted $\outcome(P)$, in which every player conforms to $\sigma_\alpha$. Concretely, $\outcome(P)=\tau_0\tau_1\cdots$ with $\tau_i = (s_i, \vec{a_i}, s_{i+1})$ for every $i\ge 0$, where $s_0$ is the initial state, and for every player $\alpha\in \Omega$ it holds that $(\vec{a_0})_\alpha=\sigma_\alpha(\epsilon)$ and $(\vec{a_j})_\alpha=\sigma_\alpha(\tau_0\cdots \tau_{j-1})$ for every $j>0$.
We denote by $\outcome_u(P)$ the outcome of the profile $P$ in the game $\Game^u$ with initial state $u$.

The cost for player $\alpha$ in profile $P$, denoted by $\cost_\alpha(P)$, is the cost it receives in $\outcome(P)$.
The cost of a profile $P$, denoted $\cost(P)$, is the tuple $\tup{\cost_\alpha(P)}_{\alpha \in \Omega}$.

Let  $P[\alpha \gets \sigma'_\alpha]$ denote the profile obtained from profile $P$ when the strategy of Player $\alpha$ 
is unilaterally changed to $\sigma'_\alpha$.
A profile is in {\em Nash equilibrium}, {\em NE} in short, if no player can obtain a lower cost by unilaterally changing its strategy. 
\begin{definition}[Nash equilibria]

A profile $ = \tup{\sigma_\alpha}_{\alpha \in \Omega}$ is said to be in Nash equilibrium if for all players $\alpha \in \Omega$, and all strategies $\sigma'_\alpha \in \strat{\alpha}$ of player $\alpha$ it holds that
$\cost_\alpha(P) \leq \cost_\alpha(P[\alpha \gets \sigma'_\alpha])$. 
\end{definition}
We say an outcome $\pi$ is in NE if there exists a NE with outcome $\pi$.
A cost vector $\vec{c}\in \mathbb{N}$ is said to be an NE if there exists an NE $P$ for which $\vec{c}=\cost(P)$. 
A cost tuple $\vec{c}$ is a \emph{Pareto-optimal Nash equilibrium} if there does not exist a NE $\vec{d} \in \mathbb{N}^\Omega$ such that 
$\vec{d}\neq \vec{c}$ and $\vec{d}\leq \vec{c}$. 
It is easy to see that a game with NE also has pareto-optimal NE. 

\paragraph{Examples}
\begin{enumerate}
	\item 
	\label{xmp: XOR game}
	[No NE]
	Figure~\ref{Fig:Eg1} represents an XOR game
	with two players, states $\set{s,t}$, actions $\set{a,b}$ for both players, initial state $s$, target set $\set{t}$ for both players. The transition function and associated costs are shown in the figure.

	It is easy to see that from every outcome of the game, one of the players can reduce their cost by flipping their actions. Hence, the game has no Nash equilibria. 
	\begin{figure*}[ht]
		\centering
			\begin{tikzpicture}[shorten >=1pt,node distance=2cm,on grid,auto, initial text =] 
			\centering
			\node[state,initial, minimum size = 1pt] (s_1)   {\footnotesize{$s$}}; 
			\node[state, accepting, minimum size = 1pt] (s_2) [right of = s_1] {\footnotesize{$t$}};
			
			\path[->] 
			(s_1)	edge  [bend left=10]  node  [align=center] {\footnotesize{$(a,a),(b,b)$}\\\footnotesize{cost = (0,1)}}  (s_2)
			(s_1)	edge  [bend right=10] node  [align=center] [below] {\footnotesize{$(a,b),(b,a)$}\\\footnotesize{cost = (1,0)}} (s_2);
			\end{tikzpicture}
			\caption{No NE}
			\label{Fig:Eg1}
	\end{figure*}
	
	
	\item 
	\label{xmp: exp many NE}
	[Exponentially many NE]
	Figure~\ref{Fig:Eg2} represents a two-player game, states $\set{s_0,\ldots,s_n,t}$, actions $\{a,b\}$ for both players, initial state $s_0$, target set $\{t\}$ for both players.
	The transition and cost functions are shown in the figure. 
	
	In this game there exists a NE with cost $(x,2^n-1-x)$ for all $0\le x\le  2^n-1$.
	If $(b_{n-1}\cdot b_0)_2$ is the binary expansion of $x$
	then Player 1 can {\em force} the outcome to take cost $(2^i,0)$ exactly when $b_i=1$, using the following strategy: 
	Both players declare that they will take action $a$ in state $s_n$ if cost $(2^i,0)$ is not taken exactly at $b_i = 1$. Taking action $a$ in $s_n$ will incur a cost of $(2^n,2^n)$, hence no agent has an incentive to deviate. 
	\begin{figure*}[ht]
		\centering
			\begin{tikzpicture}[shorten >=1pt,node distance=2cm,on grid,auto, initial text =] 
			\centering
			\node[state,initial, minimum size = 1pt] (s_1)   {\footnotesize{$s_0$}}; 
			\node[state, minimum size = 1pt] (s_2) [right of = s_1] {\footnotesize{$s_1$}};
			\node[state, minimum size = 1pt] (s_3) [right of = s_2] {\footnotesize{$s_2$}};
			\node[state, minimum size = 1pt] (s_4) [right=1.2cm of s_3] {\footnotesize{$s_{n-1}$}};
			\node[state, minimum size = 1pt] (s_5) [right=2.7cm of s_4] {\footnotesize{$s_{n}$}};
			\node[state, accepting, minimum size = 1pt] (s_6) [right of = s_5] {\footnotesize{$t$}};
			
			\path[->] 
			(s_1)	edge  [bend left=10]  node  [align=center] {\footnotesize{$(a,a),(b,b)$}\\\footnotesize{cost = (0,1)}}  (s_2)
			(s_1)	edge  [bend right=10] node  [align=center] [below] {\footnotesize{$(a,b),(b,a)$}\\\footnotesize{cost = (1,0)}} (s_2)
			(s_2)	edge  [bend left=10]  node  [align=center] {\footnotesize{$(a,a),(b,b)$}\\\footnotesize{cost = (0,2)}}  (s_3)
			(s_2)	edge  [bend right=10] node  [align=center] [below] {\footnotesize{$(a,b),(b,a)$}\\\footnotesize{cost = (2,0)}} (s_3)
			(s_3)	edge  [dashed] node  {} (s_4)
			(s_4)	edge  [bend left=10]  node  [align=center] {\footnotesize{$(a,a),(b,b)$}\\\footnotesize{cost = $(0,2^{n-1})$}}  (s_5)
			(s_4)	edge  [bend right=10] node  [align=center] [below] {\footnotesize{$(a,b),(b,a)$}\\\footnotesize{cost = $(2^{n-1},0)$}} (s_5)
			(s_5)	edge  [bend left=10]  node  [align=center] {\footnotesize{$(a,a),(a,b), (b,a)$}\\\footnotesize{cost = $(2^n,2^n)$}}  (s_6)
			(s_5)	edge  [bend right=10] node  [align=center] [below] {\footnotesize{$(b,b)$}\\\footnotesize{cost = $(0,0)$}} (s_6);
			\end{tikzpicture}
			\caption{Exponentially many NE}
			\label{Fig:Eg2}
	\end{figure*}
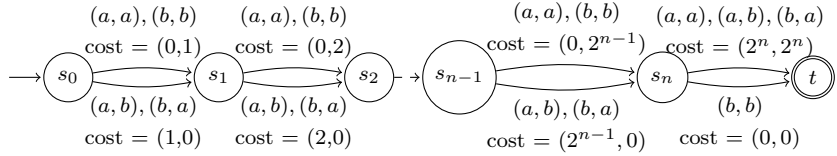
	\item 
	\label{xmp:InfNEfinitePNE}
	[Infinite NE but one Pareto-optimal NE]
	Figure~\ref{Fig:InfFinite} is a two player game with states $\{s,t,\mathsf{sink}\}$, actions $ \{a,b \}$ for both players, initial state $s$, and target set $\{t\}$ for both players. Transitions and costs are as shown in Figure~\ref{Fig:InfFinite}, and missing transitions from s go to $\mathsf{sink}$.
	
	It is easy to see that an outcome of the form $(a,a)^k(b,b)$ is a NE with cost vector $(k+1,k+1)$ for $k\geq 0$. Clearly, there is only one Pareto-optimal NE i.e. $(1,1)$. 
	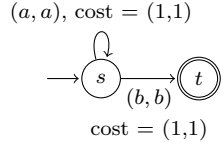
\begin{figure*}[ht]
			\centering
			\begin{tikzpicture}[shorten >=1pt,node distance=1.3cm,on grid,auto, initial text =] 
			\centering
			\node[state,initial, minimum size = 1pt] (s_1)   {\footnotesize{$s$}}; 
			\node[state, accepting, minimum size = 1pt] (s_2) [right of = s_1] {\footnotesize{$t$}};
			
			\path[->] 
			(s_1)	edge  [loop above] node  {\footnotesize{$(a,a)$, cost = (1,1)}}  (s_1)
			(s_1)	edge  node  [align=center] [below] {\footnotesize{$(b, b)$}\\\footnotesize{cost = (1,1)}} (s_2);
			\end{tikzpicture}
			\caption{Infinite NE}
			\label{Fig:InfFinite}
	\end{figure*}
\end{enumerate}



\paragraph{Problem formulation}
Examples~\ref{xmp: XOR game}-~\ref{xmp:InfNEfinitePNE} give rise to the following questions pertaining to Nash equilibria in QCGs. 

\begin{enumerate}[label={P\arabic*}]
\item {\em Existence problem}: Does a given QCG have an NE?

\item \label{Prob:Bounded} {\em Bounded Pareto-optimal NE}: Are the number of Pareto-optimal NE in a QCG bounded?

\item \label{Prob:Compute} {\em Computation problem}: If~\ref{Prob:Bounded} holds, can one compute the set of all Pareto-optimal NE cost-vectors?

\item {\em Threshold problem}: The decision version of~\ref{Prob:Compute} checks whether given a QCG and a cost vector $\vec{c}$, does there exist an NE with cost vector $\vec{d}$ such that $\vec{d}\leq\vec{c}$?
\end{enumerate}
This paper analyses each one of the above stated problems. 
For simplicity, all player actions are enabled in every state in the game
Our definitions and results remain valid when some actions may be disabled in some states. 
%

\paragraph{Representation of games}
In this paper, we distinguish between when the number of players is fixed (such as 2-player game) and when it is given as part of the input. 

The size of the transition function and cost function play a crucial role in the analysis of QCGs. Consider a game in which each player chooses from at least two actions. The number of transitions in this game is exponential in the number of players. Hence, a naive tabular representation of the transition and cost functions is exponential in size of the number of players. This encoding for the transition and cost functions leads to skewed analysis, and precludes polynomial time reductions to problems involving games with multiple players with more than one action to choose from. 

To this end, we assume that the transition function is encoded by a circuit, and in particular a model that can be efficiently queried. 
For example, consider a state $u$ from which there is a transition to state $v$ if all players play $a$, and to state $w$ otherwise. 
This is succinctly represented by the circuit implementing (\textsc{if} $\vec{a}$ \textsc{then} $v$, \textsc{else} $w$).
Our algorithms work in polynomial time in the size of these circuits, and our hardness proofs are able to output them. 
 
The cost function is also encoded using similar circuits. The representation the cost values in unary or binary can cause a difference, and hence will be explicitly mentioned. 


Finally, the remaining components of the game, namely states, actions, and accepting sets are encoded naturally as part of the input.

\section{Characterization of NE}
\label{sec: characterization of NE}
In this section we give a characterization of (Pareto optimal) NE, by showing that they are attained by strategies with a special structure. This provides intuition on the behavior of players in an NE, and forms the basis for the algorithm described in Section~\ref{sec: computing NE}.

\subsection{Game Against $\alpha$}
\label{subsec: game against m}
Consider a game $\Game = \tup{\Omega, V,  \Act, v_0, \delta,\cost,F}$ and a player $\alpha\in \Omega$. We define the \emph{game against $\alpha$}, denoted $\Game_{\widehat{\alpha}}$, to be the two-player concurrent game obtained from $\Game$ as follows. The players are $\alpha$ and the \emph{coalition} $\coal$, which comprises the set $\Omega\setminus\set{\alpha}$ of all other players. The goal of Player $\alpha$ is to minimize the cost prescribed by $\cost_\alpha$ until reaching $F_\alpha$, and the goal of the coalition is to either prevent Player $\alpha$ from reaching $F_\alpha$, or to maximize the cost prescribed by $\cost_\alpha$ until reaching $F_\alpha$.

For every state $u\in V$, let $C_\alpha(u)$ be the maximal value that can be guaranteed by the coalition in $\Game^u$. Formally, $C_\alpha(u)=\max_{\sigma\in \strat_{\Omega\setminus \set{\alpha}}}$ $\min_{\tau\in \strat_\alpha}$ $ \cost_\alpha(\outcome_u(\sigma,\tau))$.\footnote{A priori, the $\max$ should be $\sup$. However, as we shall see in Theorem~\ref{thm: compute punishing game value}, memoryless strategies suffice, and hence $C_\alpha(u)$ is always attained.}

\begin{remark}
	The reader may wonder why we look at the optimal value for the coalition, and not for Player $\alpha$. 
	Intuitively, we use the game against $\alpha$ to allow the coalition to ``punish'' Player $\alpha$ for deviating from a suggested profile (supposedly an NE). Thus, we must fix the punishing strategy for the coalition before knowing how Player $\alpha$ deviates. $C_\alpha(u)$ is then the maximal punishment against Player $\alpha$.
\end{remark}

We start by showing that $C_\alpha(u)$ is computable in polynomial time for every $\alpha\in \Omega$ and every $u\in V$.

\begin{theorem}
	\label{thm: compute punishing game value}
	Consider a game $\Game$ with costs represented in binary. $C_\alpha(u)$ is computable in polynomial time for every $\alpha\in \Omega$ and every $u\in V$. Moreover, $C_\alpha(u)$ is attained by a memoryless strategies for both players (independent of $u$). 
\end{theorem}
\begin{proof}
		For a strategy $\sigma\in \Pi_{\coal}$, denote by $\Game^\sigma$ the weighted (possibly infinite) graph obtained from $\Game$ by applying the actions prescribed by $\sigma$. The elements of $\Game$ such as $\cost$ and $F$ are naturally extended to $\Game^{\sigma}$.
	
	Observe that it suffices to prove memoryless strategies exist for the coalition. Indeed, once a memoryless strategy $\sigma$ is fixed by the coalition, the corresponding strategy for Player $\alpha$ is to choose the minimal-cost path to $F_\alpha$ (w.r.t. $\cost_\alpha$) in $\Game^{\sigma}$ (which, if $\sigma$ is memoryless, has $|V|$ states), which is clearly implemented by a memoryless strategy. 
	
	Let $U\subseteq V$ be the set of states from which the coalition can force the game never to reach $F_\alpha$. Using the results of~\cite{dAH00} on concurrent reachability games, we can compute $U$ in polynomial time, and moreover - a memoryless strategy suffices for the coalition to keep the game in $V\setminus F_\alpha$ (and clearly in this case, the strategy of Player $\alpha$ is irrelevant, as $C_\alpha(u)=\infty$ for $u\in U$). We henceforth assume that $U$ has already been computed. We remark that this assumption is not actually needed, as our algorithm will also compute this set as a by-product, but it slightly simplifies the correctness proof.
	
	We now describe an algorithm to compute $C_\alpha(u)$ for every state $u\in V\setminus U$. The algorithms stores a value $T(v)$ for every state $v\in V$, which is updated in every iteration. We refer to $T_i(v)$ as the state of $T(v)$ in iteration $i$ of the algorithm.
	
	Initially, $T_0(v)=0$ for $v\in F_\alpha$ and $T_0(v)=\infty$ otherwise.
	In every iteration, we make the following update to every state:
	\begin{equation}
	\label{eq: C algorithm update}
	T_{i+1}(v)=\max\set{\min\set{\cost_\alpha(v,a,b)+T_i(\delta(v,a,b)): a\in \Act_\alpha}:b\in \Act_{\coal}}
	\end{equation}
	
	The algorithm halts once a fixpoint has been reached, namely when $T_{i+1}\equiv T_i$. At every iteration $i$, we associate with $T_i(v)$ strategies $\mu_i\in \Pi_{\coal}$ and $\nu_i\in \Pi_\alpha$ that are obtained
	
	We now turn to prove that the algorithm terminates within $|V|$ iterations, and that upon termination, we have $T_{|V|}(v)=C_\alpha(v)$ for every state $v$.
	To this end, we prove the following inductive invariant: Let 
	\[
	S_i(v)=\max_{\sigma\in \Pi_{\coal}}\set{\min\set{\cost_\alpha(\outcome_v(\sigma,\tau)): 
			\tau\in \Pi_\alpha \text{ and }|\outcome_v(\sigma,\tau)|\le i}},
	\]
	where $|\outcome_u(\sigma,\tau)|$ is the number of transitions along the path until $F_\alpha$ is reached (and is $\infty$ is $F_\alpha$ is not reached). We claim that at iteration $i$, we have $S_i(v)=T_i(v)$.
	That is, $T_i(v)$ is the maximal value that the coalition can guarantee of a cheapest path to $F_\alpha$ of length at most $i$ (there may be longer yet cheaper paths).
	
	For $i=0$ this is trivial to observe: in $F_\alpha$ the coalition can guarantee $0$, and everywhere else $\infty$.
	Assume that the claim is correct for $i$, we prove for $i+1$. 
	For readability, in the following we always have $a\in \Act_\alpha$, $b\in \Act_{\coal}$, $\tau,\tau'\in \Pi_\alpha$, $\sigma,\sigma'\in \Pi_{\coal}$, and $v'=\delta(v,a,b)$.
	By the induction hypothesis, we have that
	\begin{align}
	&T_{i+1}(v)=\max_b \set{\min_a \set{\cost_\alpha(v,a,b)+S_i(\delta(v,a,b))}}\\
	=&\max_b \set{\min_a \set{\cost_\alpha(v,a,b)+\max_{\sigma}\set{\min_{\tau}\set{\cost_\alpha(\outcome_{v'}(\sigma,\tau)):|\outcome_{v'}(\sigma,\tau))|\le i}}}}\\
	=&\max_b \set{\min_a \set{\max_{\sigma}\set{\min_{\tau}\set{\cost_\alpha(v,a,b)+\cost_\alpha(\outcome_{v'}(\sigma,\tau)):|\outcome_{v'}(\sigma,\tau))|\le i}}}}\\
	=&\max_b \set{\max_{\sigma}\set{\min_a \set{\min_{\tau}\set{\cost_\alpha(v,a,b)+\cost_\alpha(\outcome_{v'}(\sigma,\tau)):|\outcome_{v'}(\sigma,\tau))|\le i}}}}\\
	=&\max_{\sigma'} \set{\min_{\tau'} \set{\cost_\alpha(v,\sigma'(v),\tau'(v))+\cost_\alpha(\outcome_{v'}(\sigma',\tau')):|\outcome_{v'}(\sigma',\tau'))|\le i}}\\
	=&\max_{\sigma'} \set{\min_{\tau'} \set{\cost_\alpha(\outcome_{v}(\sigma',\tau')):|\outcome_{v}(\sigma',\tau'))|\le i+1}} = S_{i+1}(v)
	\end{align}
	where the transitions are as follows:
	\begin{enumerate}
		\item[(1)-(2)] is by definition of $S_i(\delta(v,a,b))$.
		\item[(2)-(3)] is by distributivity of $\min$ and $\max$ over $+$.
		\item[(3)-(4)] is the heart of the proof. Trivially, we can write $(3)\ge (4)$ (since $\min \max$ is less than $\max \min$). 
		For the converse inequality, notice that in $(3)$ the coalition chooses a maximizing strategy $\sigma$ given the action $a$. However, $\sigma$ does not play a role in $\cost_\alpha(v,a,b)$. Therefore, the coalition may as well choose a strategy $\sigma$ that maximizes $\cost_\alpha(\outcome_{v'}(\sigma,\tau))$ for every $v'$ (which is determined by $a$). This new strategy is now independent of $a$, so we can maximize it before minimizing over $a$, as done in $(4)$, and the coalition is guaranteed not to reduce the cost.
		\item[(4)-(5)] is an aggregation of the first action with the rest of the strategy.
		\item[(5)-(6)] is by the definition of the cost of an outcome.
	\end{enumerate}
	We conclude that $S_i\equiv T_i$ for all $i$. 
	
	We now prove that the algorithm reaches a fixed point within $|V|$ iterations.
	Consider a state $v\in V\setminus U$. By definition, for every strategy $\sigma\in \Pi_{\coal}$, Player $\alpha$ has a strategy to reach $F_\alpha$ from $v$. For every such strategy $\sigma$, consider $\mu\in \arg\min_{\tau\in \Pi_{\coal}} \cost_\alpha(\outcome_u(\sigma,\tau))$, then w.l.o.g. we can assume $\outcome(\sigma,\mu)$ is a simple path (i.e. no state is visited more than once). Indeed, if a state is visited twice, then Player $\alpha$ can shorten the path without increasing the cost. Moreover, this simple path has minimal cost among all available paths from $v$ to $F_\alpha$ in $\Game^\sigma$. 
	
	We thus have that
	\begin{align*}
	&C_\alpha(v)=\max_{\sigma\in \Pi_{\coal}}\min_{\tau\in \Pi_\alpha} \cost_\alpha(\outcome_v(\sigma,\tau))\\
	&=\max_{\sigma\in \Pi_{\coal}}\set{\min\set{\cost_\alpha(\outcome_v(\sigma,\tau)): 
			\tau\in \Pi_\alpha \text{ and }|\outcome_v(\sigma,\tau)|\le |V|}}= S_{|V|}(v).
	\end{align*}
	Since $S_{|V|}(v)=T_{|V|}(v)$, we conclude that the algorithm terminates within $|V|$ iterations, and computes $C_\alpha(v)$ for every state $v\in V$. Moreover, when $T_i$ reaches a fixed point, we can extract from Equation~\ref{eq: C algorithm update} a memoryless strategy for the coalition, by choosing the maximizing action at each state.
\end{proof}
%
\subsection{Characterizing NE}
\label{subsec: characterize NE}
We are now ready to characterize Pareto-optimal NE profiles. We show that it is enough to consider strategies of a the following form: intuitively, the players agree on a short outcome and play according to it. If Player $\alpha$ deviates from the outcome, the other players form a coalition and play according to $\Game_{\widehat{\alpha}}$, as per Section~\ref{subsec: game against m}.

Let $\Game = \tup{\Omega, V,  \Act, v_0, \delta,\cost,F}$, and consider a strategy profile $P=\tup{\sigma_1,\ldots,\sigma_k}$. For every player $\alpha\in \Omega$, consider the game $\Game_{\widehat{\alpha}}$. The optimal strategy for the coalition in $\Game_{\widehat{\alpha}}$ induces a strategy $\chi^\alpha_\beta$ for every $\beta\neq \alpha$, such that the combination of these strategies forms the strategy for the coalition. 

Let $\pi=\outcome(P)=(v_0,\vec{a}_0,v_1),(v_1,\vec{a}_1,v_2),\ldots$. We define a new strategy profile $\ol{P}=\tup{\ol{\sigma_1},\ldots,\ol{\sigma_k}}$ as follows. For every $\beta\in \Omega$, as long as all other players follow $\pi$, Player $\beta$ plays according to $\sigma_i$. If, at time $i$, Player $\alpha\neq \beta$ deviates from $\pi$ such that instead of transition $(v_i,\vec{a}_i,v_{i+1})$, the transition that is taken is $(v_i,\vec{a}',v')$, then Player $\beta$ starts playing $\chi^\alpha_\beta$ from $v'$.

Clearly $\outcome(P)=\outcome(\ol{P})$. We refer to the profile $\ol{P}$ as a \emph{second--strike} profile.

\begin{lemma}
	\label{lem: second strike enough for NE}
	If $P$ is an NE, then $\ol{P}$ is also an NE.
\end{lemma}
\begin{proof}
	We prove that $\ol{P}$ is an NE by showing that no player can beneficially deviate. 
	
	Let $P=\tup{\sigma_1,\ldots,\sigma_k}$. Assume by way of contradiction that $\outcome(\ol{P})$ is not an NE. Thus, there exists some player $\alpha\in \Omega$ and a strategy $\sigma'$ for Player $\alpha$ that is a beneficial deviation from $\ol{P}$. That is, let $d=\cost_\alpha(\outcome(P))=\cost_\alpha(\outcome(\ol{P}))$ and  $\ol{d}=\cost_\alpha(\outcome(\ol{P}[m\gets \sigma']))$, then $\ol{d}<d$.
	
	Let $\pi=\outcome(P)=(v_0,\vec{a}_0,v_1),(v_1,\vec{a}_1,v_2),\ldots$, and $\ol{P}=\outcome(\ol{P}[1\gets \sigma'])=(v_0,\vec{a}_0,v_1),\ldots$ $ (v_{i-1},\vec{a}_{i-1},v_i),(v_i,\vec{b},u_{i+1}),(u_{i+1},\vec{\lambda}_{i+1},u_{i+2}),\ldots$ where $i$ is the minimal index such that $\vec{a}_i\neq \vec{b}_i$, i.e., the first time when Player $\alpha$ deviates from $\sigma_\alpha$, and $\lambda_i$ are played as per the second-strike strategies of $\Omega\setminus\{\alpha\}$.
	Define $c=\cost_\alpha((v_0,\vec{a}_0,v_1),\ldots, (v_{i-1},\vec{a}_{i-1},v_i))$ to be the cost accumulated by Player $\alpha$ along $\pi$ up to $v_{i}$, and $e=\cost_\alpha((v_i,\vec{a}_i,v_{i+1}),\ldots)$ be the cost accumulated on the suffix from $v_{i}$ along $\pi$ (recall that once $F_\alpha$ is reached, the cost does not accumulate, and is finite). Thus, $d=c+\cost_\alpha(v_{i-1},\vec{a},v_i)+e$. Similarly, let $\ol{e}=\cost_\alpha((u_{i+1},\vec{\lambda}_i,u_{i+2}),(u_{i+2},\vec{\lambda}_{i+2},u_{i+3}),\ldots$, then $\ol{d}=c+\cost_\alpha(v_i,\vec{b}_i)+\ol{e}$. Since $\ol{d}<d$, it follows that 
	\begin{equation}
	\label{eq: ole < e}
	\cost_\alpha(v_i,\vec{b}_i)+\ol{e}<\cost_\alpha(v_{i-1},\vec{a})+e.
	\end{equation}
	
	By Theorem~\ref{thm: compute punishing game value}, the second-strike strategies $\vec{\chi}=(\chi^\alpha_{\beta})_{\beta\neq \alpha}$ satisfy   $\cost_\alpha(\outcome_{u_{i+1}}(\tau,\vec{\chi}))\ge C_\alpha(u_{i+1})$ for every $\tau\in \Pi_\alpha$. In particular, we have that 
	\begin{equation}
	\label{eq: ole ge C_1}	
	\ol{e}\ge C_\alpha(u_{i+1}).
	\end{equation}
	
	Viewing $C_\alpha(u_{i+1})$ on the contrapositive, we get that for every strategy $\vec{\sigma}\in \Pi_{\Omega\setminus\set{\alpha}}$ there exists a strategy $\tau\in \Pi_\alpha$ such that $\cost_\alpha(\outcome(\tau,\vec{\sigma}))\le C_\alpha(u_{i+1})$. Let $\tau'$ be such a strategy for Player $\alpha$ against the profile $\tup{\sigma_\beta}_{\beta\neq \alpha}\in \Pi_{\Omega\setminus \set{\alpha}}$ from $u_{i+1}$, we augment $\tau'$ to the following strategy: play according to $\sigma_\alpha$ along $\pi$ up to $v_{i}$, and then play $\vec{b}$ as $\sigma'$ does. Then, proceed with $\tau'$ from $u_{i+1}$.
	We show that $\tau'$ is a beneficial deviation from $P$, in contradiction to the assumption that $P$ is an NE. 
	
	Let $\pi'=\outcome(P[\alpha\gets \tau'])=(v_0,\vec{a}_0,v_1),\ldots (v_{i-1},\vec{a}_{i-1},v_i),(v_i,\vec{b}_i,u_{i+1}),(u_{i+1},\vec{\lambda}_i,w_{i+2}),\ldots$
	where $\lambda_i$ are the actions prescribed by $P[\alpha\gets \tau']$ from $u_{i+1}$. As before, we let $d'=\cost_\alpha(\pi')$ and $e'=\cost_\alpha((u_{i+1},\vec{\lambda}_i),\ldots)$, then $d'=c+\cost_\alpha(v_i,\vec{b}_i)+e'$, with $e'\le C_\alpha(u_{i+1})$
	
	Combining this with Equations~\eqref{eq: ole < e} and \eqref{eq: ole ge C_1}, we now have
	\begin{align*}
	&d'=c+\cost_\alpha(v_i,\vec{b}_i)+e'\\
	&\le c+\cost_\alpha(v_i,\vec{b}_i)+C_\alpha(u_{i+1})\\
	&\le c+\cost_\alpha(v_i,\vec{b}_i)+\ol{e}\\
	&< c+ \cost_\alpha(v_{i-1},\vec{a})+e =d	
	\end{align*}
	and we are done.
\end{proof}
%

Consider a profile $\ol{P}$ for some (not necessarily NE) profile $P$, and a player $\alpha\in \Omega$. 
Suppose Player $\alpha$ deviates from $\ol{P}$, and that the first deviation from the outcome is the transition $(v,\vec{b},u)$.
By Lemma~\ref{lem: second strike enough for NE}, it follows that a profile is an NE iff no such player can deviate and gain more than $C_\alpha(u)$, plus the cost of the deviating edge. 


Formally, we have the following characterization.
\begin{theorem}
	\label{thm: NE deviation criteria}
	A profile $P$ with outcome $\pi=(v_0,\vec{a}_0,v_1),\ldots$ is an NE iff the following holds. 
	For every Player $\alpha\in \Omega$, for every prefix $\pi_{[0,j]}=(v_0,\vec{a}_0,v_{1}),\ldots,(v_{j-1},\vec{a}_{j-1},v_j)$ of $\pi$, and for every action $a'\in \Act_\alpha$ such that $(\vec{a}_{j})_\alpha\neq a'$, let $\vec{b}\in \Act$ be the action vector obtained from $\vec{a}_{j}$ by changing the action of Player~$\alpha$ to $a'$, and let $u=\delta(v_j,\vec{b})$, then $\cost_\alpha(\pi_{[0,j]})+\cost_\alpha(v_j,\vec{b})+C_\alpha(u)\ge \cost_\alpha(\pi)$.
\end{theorem}

Theorem~\ref{thm: NE deviation criteria}, combined with Theorem~\ref{thm: compute punishing game value} almost give us an algorithmic procedure for deciding whether a profile is an NE. Missing is a bound on the length of the outcome (until $F_\alpha$ is reached for all relevant players).  We now proceed to obtain such a bound, by bounding the memory required from Pareto-optimal NE strategies.

Let \(\Game = \tup{\Omega, V,  \Act, v_0, \delta,\cost,F}\), we obtain from $\Game$ the \emph{$F$-expanded game} $\Game^\star=\langle\Omega, V\times 2^{\Omega},  \Act, (v_0, \emptyset),$ $ \delta^\star,\cost^\star,F^\star\rangle$ as follows. We construct a copy of $\Game$ for each subset of the players. Intuitively, the subset denotes which players have already visited their target sets. Thus, the states are $V\times 2^{\Omega}$, and the initial state is $(v_0, \emptyset)$. The actions are the same as those of $\Game$. The transition function is defined as follows: for a state $(v,S)$ and action vector $\vec{a}$, we have that $\delta^\star((v,S),\vec{a})=(v',S')$ where $v'=\delta(v,\vec{a})$ and $S'=S\cup \set{\alpha\in \Omega: v\in F_\alpha}$. That is, all players who reached their target by state $v$ are added to $S$. The cost function is defined as follows: for Player $\alpha\in \Omega$ we have that 
\[\cost^\star_\alpha((v,S),\vec{a})= \begin{cases}
\cost_\alpha(v,\vec{a}) & \alpha\notin S\\
0 & \alpha\in S
\end{cases}\]
That is, once a player has reached $F_\alpha$, which is encoded in $S$, no further cost is incurred. Finally, we set $F^\star_\alpha=\set{(v,S): \alpha\in S}$. 

Clearly there is a bijection between strategies of $\Game$ and $\Game^\star$, and this induces to a bijection between outcomes, and between profiles. 

Consider a profile $P^\star$ in $\Game^\star$. We define the set of \emph{winners} $W\subseteq \Omega$ to consist of all players $\alpha$ such that $F_\alpha$ is visited along $\outcome(P^\star)$. By the construction of $\Game^\star$, $\outcome(P^\star)$ eventually reaches the copy $V\times W$, and stays there (with all players in $W$ accumulating cost $0$, and all other players incurring cost $\infty$ by definition). 

For a set $W$ of winners, we say that a transition $(x,\vec{a},y)$ (either in $\Game$ or $\Game^\star$) is \emph{safe for $W$} if for every player $\alpha\in \Omega\setminus W$ and every action $\vec{b}$ that is obtained from $\vec{a}$ by (possibly) changing the action of Player $\alpha$, the resulting transition $(x,\vec{b},z)$ satisfies $C_\alpha(z)=\infty$. 

The following is an easy observation.
\begin{lemma}
	\label{lem: safe transitions}
	\begin{enumerate}
		\item In an NE profile $P$ with winners $W$, the outcome can only take safe transitions for $W$.
		\item Consider a state $(u,W)\in \Game^\star$ as an initial state, then for a profile $P$ whose outcome remains in $V\times W$ and takes only safe transitions for $W$, the second-strike profile $\ol{P}$ is an NE.
	 \end{enumerate}
\end{lemma}

We are now ready to characterize $\NEset$ by means of the expanded game. 
\begin{lemma}
	\label{lem: simple outcome enough}
	Let $\vec{c}\in \NEset(v_0)$, then there exist an NE profile $P^\star$ in $\Game^\star$ with $\cost(P^\star)=\vec{c}$ and a set of winners $W$ such that $\outcome(P^\star)$ forms a lasso, namely a a simple path followed by a simple cycle, in $\Game^\star$.
\end{lemma}
\begin{proof}
		Let $\vec{c}\in \NEset(v_0)$, then there exists an NE profile $P$ which attains it, and moreover, by Lemma~\ref{lem: second strike enough for NE} we can assume $P$ to be a second-strike profile. Let $P^\star$ be the corresponding profile for $\Game^\star$. Let $\pi=\outcome(P^\star)$, then we can write $\pi=\mu\cdot \eta$ where $\mu$ is a maximal finite prefix of $\pi$ from $(v_0,\emptyset)$ that does not visit $V\times W$,  and $\eta$ is the infinite suffix within $V\times W$. We now modify $\pi$ to obtain a new outcome, with the desired properties, that induces an NE as per Theorem~\ref{thm: NE deviation criteria}.
		
		By Lemma~\ref{lem: safe transitions}, all transitions in $\pi$ are safe for $W$. Thus, as long as we only use transitions that are taken in $\pi$, the players in $\Omega\setminus W$ cannot gain by deviating. We henceforth focus only on the players in $W$.
		
		We consider the suffix $\eta$. Since $\eta$ is infinite and $V\times W$ is finite, then there exist cycles in $\eta$. Write $\eta=\eta_1\cdot \eta_2\cdot \eta_3$ where $\eta_1$ is a simple path, $\eta_2$ is the first simple cycle in $\eta$, and $\eta_3$ is the remaining suffix. 
		We claim that replacing $\eta$ by $\eta_1\cdot \eta_2^\omega$ induces an NE profile that attains cost $\vec{c}$. Indeed, observe that for the players in $W$, the cost does not change, and they cannot deviate by gaining, as they do not accumulate cost once $V\times W$ is reached (and since $\mu$ remains unchanged, there is no incentive to deviate).  
		
		Next, consider the prefix $\mu$. If $\mu$ contains a cycle, we claim that it can be removed: indeed, write $\mu=\mu_1\cdot \mu_2\cdot \mu_3$, where $\mu_2$ is a cycle, then the players can modify their strategies such that $\mu_1\cdot \mu_3$ is the outcome (or prefix thereof). Clearly $\cost(\mu_1\cdot \mu_3)\le \cost(\mu)$. It remains to show that this is still an NE profile. However, observe that any beneficial deviation from $\mu_1\cdot \mu_3$ induces a beneficial deviation from $\mu$, but since $\mu$ is part of an NE, this cannot exist.
		
		Since $\vec{c}$ is Pareto-optimal, there cannot be an NE profile that attains a lower cost, so we conclude that the above truncation yields cost exactly $\vec{c}$.
		
		We conclude that $\pi$ can be assumed a lasso in $\Game^\star$.
\end{proof}

By the structure of $\Game^\star$, if a state $(v',S')$ is reachable from the state $(v,S)$ in $\Game^\star$, then $S\subseteq S'$. Thus, a maximal simple path in $\Game^\star$ is of length $\Omega\times V$. It follows that the maximal length of a simple path in $\Game^\star$ is $|\Omega|\cdot |V|$. Furthermore, note that a simple cycle within $V\times W$ for some set $W$ corresponds to a simple cycle in $V$. From Lemma~\ref{lem: simple outcome enough}, we can conclude the following.

\begin{corollary}
	\label{cor: k cycle outcome enough}
	Let $\vec{c}\in \NEset(v_0)$, then it is attained by an NE profile $P$ with in $\Game$ with a set of winners $W, $such that $\outcome(P)=\mu\cdot \eta^\omega$, where $\mu$ is a path of length at most $|\Omega|\cdot |V|$ that visits $F_\alpha$ for all $\alpha\in W$, and $\eta$ is a simple cycle.
\end{corollary}

\section{Computing NE}
\label{sec: computing NE}
Combining Theorems~\ref{thm: compute punishing game value} and~\ref{thm: NE deviation criteria}, and Corollary \ref{cor: k cycle outcome enough}, gives us a simple algorithm for deciding whether $\vec{c}\in \NEset$. Given the game $\Game$, we look for a path of the form $\mu\cdot \eta$ as per Corollary~\ref{cor: k cycle outcome enough}, and check that the condition described in Theorem~\ref{thm: NE deviation criteria} holds for this path.

Note that checking the latter can be done in polynomial time, since we only need to check that (1) for deviations of all the players along the prefix $\mu$, and that (2) once we reach the cycle $\eta$, for the set of winners $W$, all the edges along $\eta$ are safe for $W$.

We thus have the following theorem.
\begin{theorem}
	\label{thm: deciding NE general}
	The problem of deciding, given a game $\Game$ and a cost vector $\vec{c}$, whether $\vec{c}\in \NEset$ is in \NP.
\end{theorem}

\subsection{Computing $\NEset$}
\label{subsec: algorithm}
A broader problem relating to NE is that of computing the entire set $\NEset(v_0)$ of Pareto-optimal NE. As it turns out, solving this problem provides insight to the effect of different parameters of the game on the complexity of the NE-with-threshold problem.

We now describe an algorithm to compute $\NEset(u)$ for every state $u$ in a game $\Game$.
Since our characterization of NE in Section~\ref{sec: characterization of NE} utilizes the $F$ - expanded graph $\Game^\star$, it will be easier to work with $\Game^\star$. Broadly, the algorithms computes for every state $u$, the set of cost vectors $\vec{c}$ from which there exists a lasso witness, as per Lemma~\ref{lem: simple outcome enough}.

The algorithm proceeds by iterating over all subsets $W\subseteq \Omega$, and computing for every state $u$ of $\Game^\star$,  the set $\NEset^W(u)$ of cost vectors $\vec{c}$ for which there exists an NE $P^\star$ with set of winners $W$. 

By Lemma~\ref{lem: safe transitions}, only safe transitions for $W$ are relevant when the set of winners is $W$. Thus, we start by computing the set of safe transitions, and removing from $\Game^\star$ all other transitions. We refer to the obtained game graph as $\Game^\star|_W$ Clearly this can be done in polynomial time in the size of $\Game^\star$. Note that in particular, every state $u$ in $\Game^\star|_W$ satisfies $C_\alpha(u)=\infty$. This means that the actions of the players in $\Omega\setminus W$ effectively do not matter, as their cost will inevitably remain $\infty$.

We assume w.l.o.g. that all states in $\Game^\star|_W$ are reachable from $(v_0,\emptyset)$, otherwise we can remove the non-reachable ones. Moreover, we assume the underlying graph has strongly connected components reachable from $(v_0,\emptyset)$, otherwise there cannot be an outcome, and we are done.

Note that once $V\times W$ is reached in $\Game^\star|_W$, there exists a cycle in $V\times W$ that satisfies condition (2) of Lemma~\ref{lem: safe transitions}, and hence induces an NE.

The algorithm stores, for every state $u\in \Game^\star|_W$, a set $D(u)$ of pairs $(i,\vec{c})$ such that $(i,\vec{c})\in D(u)$ iff there exists a path in $\Game^\star|_W$ of length at most $i$ from $u$ to $V\times W$ with cost $\vec{c}$ that satisfies the conditions of Lemma~\ref{lem: simple outcome enough}. The algorithm then iterates over the length $i$, as follows.

\paragraph*{Initilization:}  For $i=0$, we add $(0,\vec{0}_W)\in D(u)$ for every $u\in V\times W$, where $\vec{0}_W$ has $0$ in the coordinates corresponding to $W$, and $\infty$ everywhere else.
\paragraph*{Update:} At iteration $i$, we add $(i,\vec{c})$ to $D(u)$ iff the following holds. 
\begin{enumerate}
	\item There exists a transition $(u,\vec{a},v)$ in $\Game^\star|_W$ and $(j,\vec{c}')\in D(v)$ such that $j<i$ and $\vec{c}=\cost(u,\vec{a})+\vec{c}'$.
	\item For every state $w$ and every action $\vec{b}$ that differs from $\vec{a}$ only in coordinate $\alpha\in \Omega$, let $v'=\delta(u,\vec{b})$, then for every $(j,\vec{d})\in D(v)$, if $j<i$ then $(\vec{c})_\alpha\le \cost_\alpha(u,\vec{b})+\vec{d}$.
\end{enumerate}

The algorithm terminates at iteration $|V|\times |\Omega|$, and returns the cost vectors in $D(v_0)$ (or a Pareto optimal subset thereof).

Clearly, in general the complexity of the algorithm is exponential in $|\Omega|$, and is thus generally exponential. We remark later on the effect of specific parameters on the complexity.

The correctness of the algorithm is easy to prove using the results of Section~\ref{sec: characterization of NE}. By Theorem~\ref{thm: NE deviation criteria}, all the cost vectors we compute are indeed NE cost vectors. By Lemma~\ref{lem: simple outcome enough}, it is enough to check for paths up to length $|\Omega|\cdot |V|$.

\subsection{Polynomial-Time Fragments}
\label{subsec: polynomial algorithm}
As mentioned above, in general the algorithm we describe in Section~\ref{subsec: algorithm} takes exponential time. In fact, even when $|\Omega|$ is fixed, the algorithm can still take exponential time, depending on the size of $\NEset(v_0)$. We now demonstrate two cases where $\NEset$ can be computed in polynomial time. In Section~\ref{sec:hardness}, we show that the restrictions in these fragments are tight, in the sense that removing any restriction makes the problem \NPh.

\begin{theorem}
	\label{thm: fixed players unary}
	When $|\Omega|$ is fixed and $\cost$ is described in unary, computing $\NEset(v_0)$ can be done in polynomial time.
\end{theorem}
\begin{proof}
	Since $|\Omega|$ is fixed, the size of $\Game^\star$ is $O(|V|)$. By Lemma~\ref{lem: simple outcome enough}, every cost vector $\vec{c}\in \NEset(v_0)$ is attained by a profile whose outcome accumulates costs only along a path of length $|\Omega|\cdot |V|$. For unary weights, the cost along a simple path is polynomial (in fact, linear), in the length of the path.
	That is, if $M$ is the maximal cost in $\Game$, the maximal cost a player can accumulate along such a path is $M\cdot |\Omega|\cdot |V|$, which is polynomial in the description of $\Game$. It follows that $|\NEset(v_0)|\le (M\cdot |\Omega|\cdot |V|)^|\Omega|$, which is polynomial.
	
	Thus, the number of updates that are done in every iteration of the algorithm is polynomially bounded, and we conclude that the runtime of the algorithm is polynomial.
\end{proof}

In the next case we make the following restrictions: first, the game is a \emph{joint target} game, meaning that $F_\alpha=T$ is the same set for all players $\alpha\in \Omega$, and second, that the game has \emph{uniform costs}, meaning that the cost of every transition is the same, and w.l.o.g. is $\vec{1}$.

\begin{theorem}
	\label{thm: single reach uniform polynomial time}
	Computing $\NEset(v_0)$ for joint-target games with uniform costs can be done in polynomial time.
\end{theorem}
\begin{proof}
	Notice that under the premise, $\NEset(v_0)$ consists of a single cost vector $\vec{c}$, whose entries (which are all equal, since the costs are uniform) are the length of the shortest path from $v_0$ to $T$ (the joint target).
	
	Thus, computing $\NEset(v_0)$ reduces to finding the shortest path from $v_0$ to $T$, which can be done in polynomial time.
\end{proof}

\section{Hardness Results}
\label{sec:hardness}
In this section we complete the complexity picture of computing NE, by providing hardness results. Since we need a decision version of the problem, we use the most restricted version of the problem, namely deciding whether an NE exists in a game.

In Section~\ref{subsec: polynomial algorithm} we consider restrictions based on the following parameters: the number of players (fixed or not), the cost function (uniform, unary, or binary), and whether there is a single target. In this section, we provide tight hardness results to match the upper bounds in Section~\ref{sec: computing NE}. Our results are summarized in Table~\ref{tab: complexity}.

\begin{table*}[t]
	\begin{center}
	\begin{tabular}{|c|c|c|c|}
		\hline
		\backslashbox{Players}{Costs}& Uniform & Unary & Binary \\ 
		\hline 
		Fixed & \PTIME & \PTIME & \NPh (ST) \\ 
		\hline 
		Not fixed & \PTIME (ST), \NPh & \NPh (ST) & \NPh (ST) \\ 
		\hline 
	\end{tabular} 
	\end{center}
	\caption{Complexity of NE existence. (ST) stands for Single Target games. Note that hardness for a fixed number of player applies already for 2 players.}
	\label{tab: complexity}
\end{table*}

\begin{theorem}
	\label{thm: hardness 2 players binary}
	The problem of deciding whether a game has an NE is \NPh for games with 2 players and binary costs, even for single-target games.
\end{theorem}
\begin{proof}
	We show that the problem is hard by a reduction from the \NPh problem PARTITION: decide, given a set of natural numbers $\set{x_1,\ldots,x_n}$ encoded in binary, whether there exists a set $I\subseteq\set{1,\ldots,n}$ such that 
	\[
	\sum_{i\in I}x_i=\sum_{i\notin I}x_i.
	\]
	
	The main ingredient in the reduction is a component similar to Example~\ref{xmp: exp many NE}. This component consists of $n$ XOR games, where in XOR game $i$, the players incur cost of either $(0,x_i)$ or $(x_i,0)$. Thus, the players partition the numbers between them. In order to ensure that the only possible NE corresponds to an equal partition, another XOR-based component is used, which allows any player to deviate and incur a total cost of slightly more than $\sum_{i=1}^n x_i/2$. By carefully choosing the costs, this makes the only possible NE have value $\sum_{i=1}^n x_i/2$ for both players, which is possible iff there is an equal partition. We now proceed to give the detailed construction.

	Given an instance for PARTITION, we start by assuming all numbers are even (this can easily be achieved by multiplying by $2$). Let $2S=\sum_{i=1}^n x_i$, then the problem is equivalent to deciding whether there exists $I\subseteq\set{1,\ldots,n}$ such that $\sum_{i\in I}x_i=S$.
	
	We construct the following game $\Game = \tup{\Omega, V,  \Act, s, \delta,\cost,F}$. The set of players is $\Omega=\set{0,1}$. The states are $\set{v_1,\ldots,v_n}\cup\set{s,t_1,t_2,r_1,r_2}$. The actions are $\Act_0=\Act_1=\set{0,1}$. The initial state is $s$, the target sets are $F_0=F_1=\set{t_2,t_3}$. 
	We now turn to describe the transition function and the costs. 
	
	The game starts at state $s$. There, the players play a XOR game to determine whether the game proceeds to $t_1$ or to $v_1$. No costs are incurred so far.
	At $t_1$, the players again play a XOR game that reaches $t_2$ and ends the game. The costs are either $(0,1)$ or $(1,0)$. Observe that if the game proceeds along $s\to t_1\to t_2$, the outcome costs are either $(S+1,S)$ or $(S,S+1)$. Moreover, either player can swap between these costs in the XOR game.
	
	If the game proceeds to $v_1$, then the outcome goes through the sequence of states $v_1\to v_2\to\ldots v_n\to r_1\to r_2$ as follows. At each state $v_i$, the players play a XOR game that reaches $v_{i+1}$ (or $r_1$ from $v_n$). The cost of the transition is either $(x_i,0)$ or $(0,x_i)$. This is similar to the construction in Example~\ref{xmp: exp many NE}.
	Intuitively, the players decide which of them sums $x_i$, thus partitioning the numbers between them into two disjoint sets. Finally, at $r_1$, if both players agree, they proceed to $r_2$ with cost $(0,0)$, and if either of them does not agree, they proceed to $r_2$ with cost $(S+2,S+2)$, where $S+2$ can be thought of as $\infty$, as it is larger than any other possible outcome.
	
	Clearly the reduction can be done in polynomial time. We now claim that there exists a partition of the instance iff the game has an NE, and moreover -- if there exists an NE, its outcome cost is $(S,S)$. 
	
	Observe that any NE in the game must have cost at most $S$ for either player. Indeed, any player can deviate to the path $s\to t_1\to t_2$ and play the XOR game at $t_1$ to guarantee cost at most $S$.

	Conversely, no NE can have cost less than $S$ for a player. Indeed, a cost of less than $S$ can only be attained along the path to $r_2$, and by the transition, it follows that if in outcome $\pi$ we have $\cost_1(\pi)=c$, then $\cost_2(\pi)=2S-c$. Thus, if one of the players has cost less than $S$, the other player has cost more than $S$, which we showed is not an NE.
	
	Thus, if there exists an NE, it has cost $(S,S)$. It remains to show that there exists an NE iff there exists a partition.
	
	For the first direction, assume there exists a partition $I\subseteq \set{1,\ldots,n}$ such that 
	\[
	\sum_{i\in I}x_i=\sum_{i\notin I}x_i=S.
	\]
	We show that there exists an NE with cost $(S,S)$. The outcome of the NE is a path to $r_2$, where at each state $v_i$, the players play the XOR game such that the cost is $(x_i,0)$ if $i\in I$, and $(0,x_i)$ if $i\notin I$. Finally, at $r_1$, they go to $r_2$ with cost $(0,0)$. Clearly this outcome has cost $(S,S)$. Note that no player has an incentive to deviate toward $t_2$ regardless of the strategy. In order to make sure no player deviates along the path to $r_2$, we note that the second-strike strategies along the path to $r_2$ can use the transition with costs $(S+2,S+2)$ to make sure no deviation is beneficial.

	Conversely, assume there does not exist a partition of the instance, then it is easy to see that no outcome can give cost $(S,S)$, and by the above, there cannot be an NE.
\end{proof}

\begin{theorem}
	\label{thm: hardness n players unary}
	The problem of deciding whether a game has an NE is \NPh for games with unary costs, even for single-target games.
\end{theorem}
\begin{proof}
	
	The result is obtained by a reduction from 3SAT, based on a quantitative variant of the reduction in~\cite[Theorem 2]{AAK15}. For completeness, we give the complete construction.
	
	We show a reduction is from 3SAT. 	Consider a formula $\phi=c_1\wedge\ldots\wedge c_m$ over the variables $x_1,\ldots,x_n$, where each clause is of the form $c_i=(\ell^i_1\vee \ell^i_2 \vee \ell^i_3)$ with each $\ell^j_k$ being a variable or its negation.
	
	We construct a game $\Game = \tup{\Omega, V,  \Act, 1, \delta,\cost,F}$ as follows: the players are $\Omega=\set{0,\top_1,\bot_1,\ldots,\top_n,\bot_n}$. That is, each variable $x_i$ is associated with two players, $\top_i$ and $\bot_i$, and there is a special Player $0$.
	
	The states are
	\begin{align*}
	V&=\{1,\ldots,m,m+1\}\cup\{\top_1,\bot_1,\ldots,\top_n,\bot_n\}\\
	&\cup\{\tup{x_i,j},\tup{\neg x_i,j}: 1\le i\le n,\ 1\le j\le m\}.
	\end{align*}
	Note that $m+1$ is a special state that does not correspond to a clause. 
	The target sets are $\set{\top_1,\bot_1,\ldots,\top_n,\bot_n,m+1}$ for all players. 
	
	We now turn to describe the transitions, actions, and costs. See Figure
	\todo{figure}

	The game starts at state $1$, corresponding to clause $c_1$. 
	\begin{itemize}
		\item In state $j$, for $1\le j\le m$, Player 0 chooses a state $\tup{\ell,j}$ for a literal $\ell$ that appears in $c_j$. That is, Player $0$ has 3 choices.\footnote{Note that our definition of $V$ includes states that are not reachable, namely literals that do not appear in the clause.}
		The actions of all the other players are ignored. The cost of this transition is $0$ for all players.
		
		\item In state $\tup{x_i,j}$, Player $0$ and Player $\top_i$ play a XOR game that can go either to state $j+1$ or to $\top_i$. Intuitively, Player $\top_i$ can either let Player 0 continue to the next clause, or ``stop'' the game and go to $\top_i$.
		
		The cost of the transition to $\top_i$ is $1$ for Player $\top_i$ and for Player $0$, and is $0$ for all other Players. Intuitively, Player $\top_i$ has to pay cost of $1$ for causing the game to deviate, and Player $0$ incurs a cost of $1$ for this deviation.
		
		The cost of the transition to $j+1$ is $2$ for Player $\bot_i$ and $0$ for all other players. Intuitively, $\bot_i$ is penalized, since $x_i$ was chosen to be true. Note that there is no deviation Player $\bot_i$ can take at this point.
		
		\item Dually, in state $\tup{\neg x_i,j}$ Player $0$ and Player $\bot_i$ play a XOR game that can go either to state $j+1$ or to $\bot_i$. Again, the cost of the transition to $\bot_i$ is $1$ for Player $\bot_i$ and for Player $0$, and is $0$ for all other Players, and the cost for the transition to $j+1$ has cost $2$ for $\top_i$.
	\end{itemize}
	
	We claim that $\phi$ is satisfiable iff there exists an NE in the game.
	
	For the first direction, assume $\phi$ is satisfiable, and let $\pi$ be a satisfying assignment (i.e., $\pi:\set{x_1,\ldots,x_n}\to \set{\top,\bot}$). We construct strategies for the Players as follows: At each state $j$ for $1\le j\le m$, Player $0$ chooses a literal $\ell$ that is assigned to true in $\pi$. Then, at each state of the form $\tup{x_i,j}$ or $\tup{\neg x_i,j}$, Player $0$ continues to state $j+1$. The strategies of the other players are such that they cooperate and never cause the game to end up in $\top_i$ or $\bot_i$. We claim this is an NE. First, observe that in this profile, Player $0$ has cost $0$, which is optimal. For the other players, consider a variable $x_i$. Since $\pi$ is a consistent assignment, then it cannot be the case that both $\tup{x_i,j}$ and $\tup{\neg x_i,j'}$ are visited, since one of them is false, and will not be chosen by Player $0$ in this profile. W.l.o.g. assume $\pi(x_i)=\top$, then $\tup{\neg x_i,j}$ is never visited for any $1\le j\le m$. This means that $\bot_i$ never gets to influence the game. In addition, the cost for Player $\top_i$ is 0, since no state of the form $\tup{\neg x_i,j'}$ is visited, and since Player $\top_i$ never deviates to state $\top_i$. Since a cost of $0$ is optimal, we conclude that the profile is an NE.
	
	Conversely, assume there exists an NE in the game. We claim that $\phi$ is satisfiable.
	
	Assume by way of contradiction that for some variable $x_i$, both $\tup{x_i,j}$ and $\tup{\neg x_i,j'}$ are visited for some $1\le j, j'\le m$. W.l.o.g. assume $j< j'$ (the case where $j'<j$ is symmetric). Consider the cost incurred by Player $\top_i$. Since $\tup{\neg x_i,j'}$ is visited after $\tup{x_i,j}$, it follows that Player $\top_i$ does not stop the game at $\tup{x_i,j}$. However, Player $\top_i$ incurs a cost of 2 when exiting $\tup{\neg x_i,j'}$. Therefore, if the profile is fixed, Player $\top_i$ can gain by deviating at $\tup{x_i,j}$ and stopping the game, and paying 1 instead of 2. 
	
	Thus, if the profile causes the game to end in state $m+1$, then it induces a consistent assignment. Moreover, since Player $0$ can only choose satisfying literals in each clause, we get that the induced assignment is satisfying. 
	
	It remains to show that the game does end in $m+1$. Indeed, otherwise some player $\top_i$ or $\bot_i$ stopped the game at some point. However, in this case Player $0$ has cost of $1$, and can deviate unilaterally to reroute the game back to $m+1$, so the profile is not an NE, in contradiction to the assumption. 
	
	This concludes the proof.
\end{proof}


\begin{theorem}
	The problem of deciding whether a game has an NE is \NPh for games with uniform costs.
	\label{thm: hardness uniform costs}
\end{theorem}
\begin{proof}
	We show a reduction from the \NPh problem HAMPATH: given a directed graph $G$ and a designated vertex $s$, decide whether there exists a Hamiltonian path starting from $s$ in $G$.
	
	Intuitively, given $G=\tup{V,E}$ where $V$ is a set of $n$ vertices and $E\subseteq V\times V$ are the edges, and given $s\in V$, the output of the reduction is a game $\Game$ whose players are $V$ and whose states contain $V$, as well as additional components. The target set of player $V$ contains the vertex $v$ (and some additional states of the game). Intuitively, the game proceeds as follows: if all the players agree on a path through the graph from $s$, then that path is taken. However, any player can deviate (and effectively ``stop'' the game) by incurring an overall cost of slightly more than $n$. In addition, we construct the game such that the only possible NE must induce a path through $G$ from $s$.
	
	If there is a Hamiltonian path from $s$, all players incur cost of at most $n$ on this path, so it is an NE. 
	Otherwise, every path either repeats a vertex, or misses a vertex. 
	In the former case, a player whose vertex is missed can surely gain by deviating and stopping the game. If no vertices are missed, then some vertex is visited for the first time after more than $n$ transitions. By carefully specifying the costs, we ensure that the corresponding player benefits by deviating and stopping the game.
	
	We proceed to give the complete details.
	
	Given $G=\tup{V,E}$ where $V$ is a set of $n$ vertices and $E\subseteq V\times V$ are the edges, and given $s\in V$, we construct a game $\Game = \tup{\Omega, S,  \Act, v_0, \delta,\cost,F}$ as follows. The players are the vertices of the graph, $\Omega=V$. The states of the game are $S=V\cup E\cup \set{q_0,q_1,\ldots,q_{2n+1}}\cup (\set{r_0,r_1,\ldots r_{2n+3}}\times V)$. The costs are uniformly $1$ on all transitions. The target sets are as follows: for every $v\in V=\Omega$, $F_v=\set{v,q_{2n-1},(r_{2n+3},v)}$.
	
	We now turn to describe the transitions and actions.
	\begin{itemize}
		\item The game starts at $q_0$. There, the players play an $n$-way XOR game to choose whether to proceed to $s$ (the designated vertex in the graph) or to $q_0$.
		\item From $q_0$ the game proceeds along the path $q_0\to q_1\to\ldots q_{2n+1}$ regardless of the actions. Note that the cost to all players in this outcome is $2n+1$ (as there are $2n+1$ transitions from $q_0$ to $q_{2n+1}$).
		\item For every state $v\in V$ (and in particular for $s$), each player chooses an edge $e=(v,u)\in E$ for some $u\in V$. If all players agree on the same edge, the game proceeds to state $e$, and then to $u$ regardless of the actions. 
		
		Otherwise, intuitively, each player $v\in V$ can cause the game to proceed to $(r_0,v)$. Formally, the action of each Player $v$ prescribes a number $f_v$, and the game proceeds to state $(r_0,f)$ where $f=\sum_{v\in V} f_v \mod n$ under an arbitrary enumeration of $V$. 
			
		\item From state $(r_0,v)$, the game proceeds along the path $(r_0,v)\to (r_1,v)\to\ldots (r_{2n+3},v)$ regardless of the actions. Note that since $(r_{2n+3},v)\in F_v\setminus\bigcup_{v'\neq v} F_{v'}$, then only Player~$v$ incurs a finite cost, and this cost is at least $2n+3$. 
	\end{itemize}

	We now claim that the game has an NE iff there exists a Hamiltonian path in $G$ from $s$. For the first direction, assume there exists a Hamiltonian path $s_1,s_2,\ldots,s_n$ in $G$ with $s=s_1$. Consider strategies for the players whose outcome is the path. That is, the vertices along the outcome are $q_0,s_1,(s_1,s_2),s_2,(s_2,s_3),\ldots (s_{n-1},s_n),s_n$. The length of this path is $2n$, and since this is a Hamiltonian path, each vertex is visited along it. So the cost for every player is at most $2n$. It is easy to observe that any deviation for any player would give cost of at least $2n+1$ to that player, so this is an NE.
	
	Conversely, consider an NE in the game. We claim that its outcome must describe a Hamiltonian path in $G$. First, notice that any outcome starting with a transition from $q_0$ to $q_1$ cannot be an NE, since Player $s\in V$ can always deviate and cause the play to reach $s$ after one step, thus getting cost $1$. 
	
	Second, the outcome cannot proceed to vertex $(r_0,v)$ for any Player $v$ - indeed, this causes all other players to incur cost $\infty$, in which case any one of them can deviate at $q_0$ and take the outcome to $q_{2n+1}$.
	
	Thus, the outcome must induce a path in $G$. Next, observe that all vertices must be visited along this path, as otherwise a player whose corresponding vertex is not visited (and hence incurs cost $\infty$) can deviate at $q_0$ to take the outcome to $q_{2n+1}$. Finally, we claim that all vertices are visited within the first $n$ vertices (i.e. within the first $2n$ transitions). Indeed - if vertex $v$ is visited after more than $2n$ transitions, then since every transition within $V$ includes an ``edge step'', the cost of Player $v$ is at least $2n+2$, in which case Player $v$ can deviate at $q_0$ to take the outcome to $q_{2n+1}$. 
	
	We conclude that the outcome of an NE induces a Hamiltonian path from $s$ in $G$.	
\end{proof}

\section{Price of Stability and Price of Anarchy}
\label{subsec: PoS}
Decentralized decision-making, in the form of concurrency, may lead to sub-optimal solutions from the point of view of society as a whole. This sub-optimality can be quantified by the concepts of \emph{price of stability (PoS)} and \emph{price of anarchy (PoA)}~\cite{ADKTWR08}, which we study in this section.

For a cost vector $\vec{c}=(c_\alpha)_{\alpha\in \Omega}$, we define the \emph{social utility} $\util(\vec{c})=\sum_{\alpha\in \Omega} c_\alpha$ to be the sum of the costs. We then define for a game $\Game$ the \emph{social optimum} $\SO(\Game)=\min_{P\in \strat_{\Omega}}\util(\cost(P))$ as the minimal possible social utility that can be attained in $\Game$. We assume w.l.o.g. $\SO(\Game)\ge 1$. This can be achieved by enforcing an initial transition with cost $1$ for all players.

Intuitively, the social optimum captures the value in case of a centralized authority. Let $\Upsilon$ denote the set of NE cost vectors in $\Game$, we then define $\PoS(\Game)=\frac{\min_{\vec{c}\in \Upsilon} \util(c)}{\SO(\Game)}$ and $\PoA(\Game)=\frac{\sup{\vec{c}\in \Upsilon} \util(c)}{\SO(\Game)}$. Intuitively, $\PoS(\Game)$ measures how much society losses from the lack of a centralized authority, under the assumptions that player will collaborate in a suggested NE (hence taking the ``best'', or minimal, NE). $\PoA(\Game)$ does not assume any collaboration, and so takes into account the ``worst'' NE. Clearly $1\le \PoS(\Game)\le \PoA(\Game)$, and the closer these values are to $1$, the more ``stable'' the game is.

As the following example shows, concurrent games with costs are not stable: $PoS(\Game)$ may be exponentially large, and $\PoA(\Game)$ may be infinite (even in games where NE exist).  Corollary~\ref{cor: k cycle outcome enough} in Section~\ref{sec: characterization of NE} implies that this bound is tight -- $\PoS(\Game)$ cannot be bigger than exponential.

\begin{example}
	Consider the 2-player game $\Game$ in Figure~\ref{fig: PoS}, which proceeds as follows. The actions are $\set{0,1}$ for both players. At state $s_0$, if both players choose $0$, the game proceeds to state $s_1$, where a XOR game is played to reach $s_2$ with cost either $(0,1)$ or $(1,0)$, and $s_2$ is a target for both players.
	
	\begin{figure}[ht]
		\centering
		\includegraphics[width=0.5\linewidth]{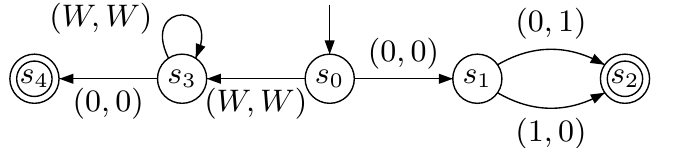}
		\caption{A game with exponential PoS and infinite PoA.}
		\label{fig: PoS}
	\end{figure}
	
	If any player plays $1$ at $s_0$, a cost of $(W,W)$ is incurred, for some number $W$ described in binary, and the game proceeds to state $s_3$. At state $s_3$, if both player play $0$, the game proceeds to $s_4$, which is a target for both player. 
	Otherwise, the game stays at $s_3$, and incurs another $(W,W)$ cost. 
	
	It is not hard to see that $\SO(\Game)=1$, by selecting some branch to $s_2$. However, the best NE is induced by the outcome $s_0,s_3,s_4$, which incurs cost $(W,W)$, and there are arbitrarily bad NEs attained by both players playing $1$ in state $s_3$. Thus, $\PoS(\Game)=2W$ and $\PoA(\Game)=\infty$.	
\end{example}
\section*{Related Work}
After preparing this paper, we discovered that the results overlap with~\cite{KLST12}. Our contribution differs by giving a polynomial-time algorithm for computing Pareto Optimal NE in certain fragments (Section~\ref{subsec: polynomial algorithm}), by giving refined hardness bounds (Section~\ref{sec:hardness}), and by studying the Price of Stability and Price of Anarchy (Section~\ref{subsec: PoS}).

\bibliographystyle{plain}
\bibliography{refs}

\end{document}